\documentclass[twoside,english]{iopart}
\usepackage[T1]{fontenc}
\usepackage[latin9]{inputenc}
\usepackage{geometry}
\expandafter\let\csname equation*\endcsname\relax
\expandafter\let\csname endequation*\endcsname\relax 

\usepackage{amsthm}
\usepackage{amsmath}
\usepackage{amssymb}
\usepackage{esint}

\makeatletter
 
 \@ifundefined{textcolor}{}
 {%
   \definecolor{BLACK}{gray}{0}
   \definecolor{WHITE}{gray}{1}
   \definecolor{RED}{rgb}{1,0,0}
   \definecolor{GREEN}{rgb}{0,1,0}
   \definecolor{BLUE}{rgb}{0,0,1}
   \definecolor{CYAN}{cmyk}{1,0,0,0}
   \definecolor{MAGENTA}{cmyk}{0,1,0,0}
   \definecolor{YELLOW}{cmyk}{0,0,1,0}
 }
  \theoremstyle{definition}
  \newtheorem{defn}{\protect\definitionname}
  \theoremstyle{plain}
  \newtheorem{prop}{\protect\propositionname}
\theoremstyle{plain}
\newtheorem{thm}{\protect\theoremname}
  \theoremstyle{plain}
  \newtheorem{cor}{\protect\corollaryname}

\makeatother

\usepackage{babel}
  \providecommand{\definitionname}{Definition}
  \providecommand{\propositionname}{Proposition}
\providecommand{\corollaryname}{Corollary}
\providecommand{\theoremname}{Theorem}

\begin{document}



\title[Generalized $q$-Cramér-Rao inequalities] {On generalized Cramér-Rao inequalities, generalized Fisher informations
and characterizations of generalized $q$-Gaussian distributions\footnote{This is a preprint version that differs from the published version, J. Phys. A: Math. Theor. 45 255303 2012, doi:10.1088/1751-8113/45/25/255303, in minor revisions, pagination and typographics details.} }

\author{J.-F. Bercher}

\ead{jean-francois.bercher@univ-paris-est.fr}


\address{Laboratoire d'informatique Gaspard Monge, UMR 8049  ESIEE-Paris,
Université Paris-Est}

\date{\today}
\begin{abstract}
This paper deals with Cramér-Rao inequalities in the context of nonextensive
statistics and in estimation
theory. It gives characterizations of generalized $q$-Gaussian distributions,
and introduces generalized versions of Fisher information. The contributions
of this paper are (i) the derivation of new extended Cramér-Rao inequalities
for the estimation of a parameter, involving general $q$-moments
of the estimation error, (ii) the derivation of Cramér-Rao inequalities
saturated by generalized $q$-Gaussian distributions, (iii) the definition
of generalized Fisher informations, (iv) the identification and interpretation
of some prior results, and finally, (v) the suggestion of new estimation
methods.

\end{abstract}





\maketitle

\section{Introduction}

It is well known that the Gaussian distribution has a central role
with respect to classical information measures and inequalities. For
instance, the Gaussian distribution maximizes the entropy over all
distributions with the same variance; see \cite[Lemma 5]{dembo_information_1991}.
Similarly, the Cramér-Rao inequality, e.g. \cite[Theorem 20]{dembo_information_1991},
shows that the minimum of the Fisher information over all distributions
with a given variance is attained for the Gaussian distribution. Generalized
$q$-Gaussian distributions arise as the maximum entropy solution
in Tsallis' nonextensive thermostatistics, which is based on the use
of the generalized Tsallis entropy. The Generalized $q$-Gaussian
distributions also appear in other fields, namely as the solution
of non-linear diffusion equations, or as the distributions that saturate
some sharp inequalities in functional analysis. Furthermore, the generalized
$q$-Gaussian distributions form a versatile family that can describe
problems with compact support as well as problems with heavy tailed
distributions. 

Since the standard Gaussian is both a maximum entropy and a minimum
Fisher information distribution over all distributions with a given
variance, a natural question is to find whether this can be extended
to include the case of the generalized $q$-gaussian distribution,
thus improving the information theoretic characterization of these
generalized $q$-Gaussians. This question amounts to look for a definition
of a generalized Fisher information, that should include the standard
one as a particular case, and whose minimum over all distributions
with a given variance is a $q$-Gaussian. More generally, this should
lead to an extension of the Cramér-Rao inequality saturated by $q$-Gaussian
distributions. 

Several extensions of the Fisher information and of the Cramér-Rao
inequality have been proposed in the literature. In particular, the
beautiful work by Lutwak, Yang and Zhang \cite{lutwak_Cramer_2005} gives an extended Fisher
information and a Cramér-Rao inequality saturated by $q$-Gaussian
distributions. In an interesting work, Furuichi \cite{furuichi_maximum_2009}
defines another generalized Fisher information and a Cramér-Rao inequality
saturated by $q$-Gaussian distributions. In this latter contribution,
the statistical expectations are computed with respect to escort distributions.
These escort distributions are one parameter deformed versions of
the original distributions, and are very useful in different formulations
of nonextensive statistics. In the present paper, we will recover
these results and show that the two generalized Fisher informations
above and the associated Cramér-Rao inequalities are actually linked
by a simple transformation. 

In section \ref{sec:Definitions-and-main}, we give the main definitions
and describe the ingredients that are used in the paper. In particular,
we give the definition and describe the importance of generalized
$q$-Gaussians, we define the notion of escort distributions, and
finally give some definitions related to deformed calculus in nonextensive
statistics. 

The Cramér-Rao inequalities indicated above are inequalities characterizing
the probability distribution, and the Fisher information is the Fisher
information of the distribution. Actually, the Fisher information
is defined in a broader context as the information about a parameter
of a parametric family of distributions. In the special case of a
location parameter, it reduces to the Fisher information of the distribution.
The Cramér-Rao inequality appears in the context of estimation theory,
and as it is well known, defines a lower bound on the variance of
any estimator of a parameter. In section \ref{sec:Generalized-Cram=0000E9r-Rao-inequalit},
we describe the problem of estimation, recall the classical Cramér-Rao
inequality, and show that the standard Cramér-Rao inequality can be
extended in two directions. First we consider moments of any order
of the estimation error, and second we use generalized moments computed
with respect to an escort distributions. This lead us to two new Cramér-Rao
inequalities for a general parameter, together with their equality
conditions. In this context, the definitions of the generalized Fisher
informations will pop up very naturally. For a very general definition
of escort distributions, we also recover a general Cramér-Rao inequality
given in a deep paper by Naudts \cite{naudts_estimators_2004}. 

In section \ref{sec:Inequalities-in-the}, we examine the special
case of a translation parameter and show, as an immediate consequence,
that the general results enable to easily recover two Cramér-Rao inequalities
that characterize the generalized $q$-Gaussian distributions. So
doing, we recover the definitions of generalized Fisher informations
previously introduced by Lutwak \textit{et al.} \cite{lutwak_Cramer_2005} and Furuichi \cite{furuichi_maximum_2009}.
Furthermore, we show that the related Cramér-Rao inequalities, which
are similar to those of \cite{lutwak_Cramer_2005} and \cite{furuichi_maximum_2009},
are saturated by the generalized $q$-Gaussians. Finally, in section
\ref{sec:Further-remarks}, we discuss some new estimation rules that
emerge from this setting, and in particular point out a connection to
the MLq-likelihood method that has been introduced recently, c.f.
\cite{ferrari_maximum_2010,hasegawa_properties_2009}.

\section{\label{sec:Definitions-and-main}Definitions and main ingredients}

\subsection{Generalized $q$-Gaussian}

The generalized Gaussian distribution is a family of distributions
which includes the standard Gaussian as a special case. These generalized
Gaussians appear in statistical physics, where they are the maximum
entropy distributions of the nonextensive thermostatistics \cite{tsallis_introduction_2009}.
In this context, these distributions have been observed to present
a significant agreement with experimental data. They are also analytical
solutions of actual physical problems, see \cite{lutz_anomalous_2003,schwaemmle_q-gaussians_2008}
\cite{vignat_isdetection_2009}, \cite{ohara_information_2010},
and are sometimes known as Barenblatt-Pattle functions, following
their identification by \cite{barenblatt_unsteady_1952,pattle_diffusion_1959}.
Let us also mention that the generalized Gaussians are the one dimensional
versions of explicit extremal functions of Sobolev, log-Sobolev or
Gagliardo\textendash{}Nirenberg inequalities on $\mathbb{R}^{n},$
as was shown by \cite{del_pino_best_2002,del_pino_optimal_2003}
for $n\geq2$ and by \cite{agueh_sharp_2008} for $n\geq1$. 
\begin{defn}
 Let $x$ be a random variable on $\mathbb{R}.$ For $\alpha\in(0,\infty),$
$\gamma$ a real positive parameter and $q>1-\alpha,$ the generalized
$q$-Gaussian with scale parameter $\gamma$ has the symmetric probability
density 
\begin{equation}
G_{\gamma}(x)=\begin{cases}
\frac{1}{Z(\gamma)}\left(1-\left(q-1\right)\gamma|x|^{\alpha}\right)_{+}^{\frac{1}{q-1}} & \text{for }q\not=1\\
\frac{1}{Z(\gamma)}\exp\left(-\gamma|x|^{\alpha}\right) & \text{if }q=1
\end{cases}\text{ }\label{eq:qgauss_general}
\end{equation}
where we use the notation $\left(x\right)_{+}=\mbox{max}\left\{ x,0\right\} $,
and where $Z(\gamma)$ is the partition function such that $G_{\gamma}(x)$
integrates to one: 
\begin{equation}
Z(\gamma)=\frac{2}{\alpha}\left(\gamma\right)^{-\frac{1}{\alpha}}\times\begin{cases}
(1-q)^{-\frac{1}{\alpha}}B\left(\frac{1}{\alpha},-\frac{1}{q-1}-\frac{1}{\alpha}\right) & \text{for }1-\alpha<q<1\\
(q-1)^{-\frac{1}{\alpha}}B\left(\frac{1}{\alpha},\frac{1}{q-1}+1\right) & \text{for }q>1\\
\Gamma\left(\frac{1}{\alpha}\right) & \text{if }q=1.
\end{cases}\label{eq:GenPartitionFunction}
\end{equation}
where $B(x,y)$ is the Beta function. 
\end{defn}
For $q>1$, the density has a compact support, while for $q\leq1$
it is defined on the whole real line and behaves as a power distribution
for $|x|\rightarrow\infty.$ Notice that the name generalized Gaussian
is sometimes restricted to the case $q=1$ above. In this case, the
standard Gaussian is recovered with $\alpha=2.$

\subsection{Escort distributions}

The escort distributions are an essential ingredient in the nonextensive
statistics context. Actually, the escort distributions have been introduced
as an operational tool in the context of multifractals, c.f. \cite{chhabra_direct_1989},
\cite{beck_thermodynamics_1993}, with interesting connections with
the standard thermodynamics. Discussion of their geometric properties
can be found in \cite{abe_geometry_2003,ohara_dually_2010}. Escort
distributions also prove useful in source coding, as noticed in \cite{bercher_source_2009}.
They are defined as follows.

If $f(x)$ is a univariate probability density, then its escort distribution
$g(x)$ of order $q$, $q\geq0,$ is defined by
\begin{equation}
g(x)=\frac{f(x)^{q}}{\int f(x)^{q}\mathrm{d}x},\label{eq:escort_f-1}
\end{equation}
provided that Golomb's ``information generating
function''   \cite{golomb_information_1966}
\begin{equation}
M_{q}[f]=\int f(x)^{q}\text{d}x\label{eq:InfoGeneratingFunctionDef}
\end{equation}
is finite.

Given that $g(x)$ is the escort of $f(x)$, we see that $f(x)$ is
itself the escort of $g(x)$ of order $\bar{q}=1/q$. 

Accordingly, the (absolute) generalized $q$-moment of order $p$
is defined by
\begin{equation}
m_{p,q}[f]:=E_{q}\left[|x|^{p}\right]=\int|x|^{p}g(x)\text{d}x=\frac{\int|x|^{p}f(x)^{q}\text{d}x}{\int f(x)^{q}\text{d}x},\label{eq:DefGeneralizedMoment}
\end{equation}
where $E_{q}[.]$ denotes the statistical expectation with respect
to the escort of order $q$. Of course, standard moments are recovered
in the case $q=1.$

\subsection{Deformed functions and algebra}

In Tsallis statistics, it has appeared convenient to use deformed
algebra and calculus, c.f. \cite{nivanen_generalized_2003,borges_possible_2004}.
The $q$-exponential function is defined by
\begin{equation}
\exp_{q}(x):=\left(1+(1-q)x\right)_{+}^{\frac{1}{1-q}},\label{eq:defExpq}
\end{equation}
while its inverse function, the $q$-logarithm is defined by
\begin{equation}
\ln_{q}(x):=\frac{x^{1-q}-1}{1-q}.\label{eq:defLnq}
\end{equation}
When $q$ tends to 1, both quantities reduce to the standard functions
$\exp(x)$ and $\ln(x)$ respectively. In the following, we will use
the notation $\bar{q}=1/q$ that already appeared above in connection
with escort distributions, and the notation $q_{*}=2-q$ that changes
the quantity $(1-q_{*})$ into $(q-1),$ e.g. $\exp_{q_{*}}(x):=\left(1+(q-1)x\right)_{+}^{\frac{1}{q-1}}.$
We note the following expressions for the derivatives of deformed
logarithms: 
\begin{equation}
\frac{\partial}{\partial\theta}\ln_{q}\left(f(x;\theta)\right)=\frac{\frac{\partial}{\partial\theta}f(x;\theta)}{f(x;\theta)}f(x;\theta)^{1-q}\text{ \,\,\,\ and \,\,\,}\frac{\partial}{\partial\theta}\ln_{q_{*}}\left(f(x;\theta)\right)=\frac{\frac{\partial}{\partial\theta}f(x;\theta)}{f(x;\theta)}f(x;\theta)^{q-1}.\label{eq:derivLnq}
\end{equation}

The $q$-product is a deformed version of the standard product such
that standard properties of exponential and logarithm functions still
hold for their deformed versions. The $q$-product is defined by
\begin{equation}
x\otimes_{q}y:=\left(x^{1-q}+y^{1-q}-1\right)^{\frac{1}{1-q}}
\end{equation}
and gives the identities 
\begin{equation}
\ln_{q}(x\otimes_{q}y)=\ln_{q}(x)+\ln_{q}(x)\text{ \,\,\ and\,\,\ }\exp_{q}(x+y)=\exp_{q}(x)\otimes_{q}\exp_{q}(y)\label{eq:qproductidentities}
\end{equation}

\subsection{Fisher information}

The importance of Fisher information as a measure of the information
about a parameter in a distribution is well known, as exemplified
in estimation theory by the Cramér-Rao bound which provides a fundamental
lower bound on the variance of an estimator. The statement of the
standard Cramér-Rao inequality, as well as several extensions, will
be given in section\,\ref{sec:Generalized-Cram=0000E9r-Rao-inequalit}. 

It might be also useful to note that Fisher information is used as
a method of inference and understanding in statistical physics and
biology, as promoted by Frieden \cite{frieden_physicsfisher_2000,frieden_sciencefisher_2004}.
It is also used as a tool for characterizing complex signals or systems,
with applications, e.g. in geophysics, in biology, in reconstruction
or in signal processing. Information theoretic inequalities involving
Fisher information have attracted lot of attention for characterizing
statistical systems through their localization in information planes,
e.g. the Fisher-Shannon information plane \cite{vignat_analysis_2003,romera_Fisher-Shannon_2004}
or the Cramér-Rao information plane \cite{dehesa_Cramer-Rao_2006}.
 
\begin{defn}
Let $f(x;\theta)$ denote a probability density defined over a subset
$X$ of $\mathbb{R}$, and $\theta\in\Theta$  a real parameter.
Suppose that $f(x;\theta)$ is differentiable with respect to $\theta$.
Then, the Fisher information in the density $f$ about the parameter
$\theta$ is defined as
\begin{equation}
I_{2,1}[f,\theta]=\int_{X}\left(\frac{\partial\ln f(x;\theta)}{\partial\theta}\right)^{2}f(x;\theta)\mathrm{d}x.\label{eq:GenFisher-1}
\end{equation}
When $\theta$ is the location parameter, i.e. $f(x;\theta)=f(x-\theta)$,
the Fisher information, expressed at $\theta=0,$ becomes a characteristic
of the distribution: the Fisher information of the distribution:
\begin{equation}
I_{2,1}[f]=\int_{X}\left(\frac{\text{d\,}\ln f(x)}{\text{d}x}\right)^{2}f(x)\mathrm{d}x.
\end{equation}
 The meaning of the subscripts in the definition will appear in the
following. 
\end{defn}

\section{\label{sec:Generalized-Cram=0000E9r-Rao-inequalit}Generalized Cramér-Rao
inequalities}

In this section, we begin by recalling the context of estimation,
the role of Fisher information and the statement of the standard Cramér-Rao
theorem. Then, we show how this can be extended to higher moments,
and to generalized moments computed with respect to an escort distribution.

\subsection{The standard Cramér-Rao inequality}

The problem of estimation, in a few words, consists in finding a function
$\hat{\theta}(x)$ of the data $x,$ that approaches the unknown value
of a characteristic parameter $\theta$ (e.g. location, scale or shape
parameter) of the probability density of these data. 

A standard statement of the Cramér-Rao inequality is recalled now.
\begin{prop}
{[}Cramér-Rao inequality{]} Let $f(x;\theta)$ be a univariate probability
density function defined over a subset $X$ of $\mathbb{R}$, and
$\theta\in\Theta$ a parameter of the density. If $f(x)$ is continuously
differentiable with respect to $\theta,$ satisfies some regularity
conditions that enable to interchange integration with respect to
$x$ and differentiation with respect to $\theta,$ then for any estimator
$\hat{\theta}(x)$ of the parameter $\theta,$ 
\begin{gather}
E\left[\left|\hat{\theta}(x)-\theta\right|^{2}\right]I_{2,1}[f;\theta]\geq\left|1+\frac{\partial}{\partial\theta}E\left[\hat{\theta}(x)-\theta\right]\right|^{2}.\label{eq:StandardCramerRao}
\end{gather}
When the estimator is unbiased, that is if $E\left[\hat{\theta}(x)\right]=\theta,$
then the inequality reduces to 
\begin{equation}
E\left[\left|\hat{\theta}(x)-\theta\right|^{2}\right]I_{2,1}[f]\geq1.
\end{equation}
 
\end{prop}
The estimator is said efficient if it is unbiased and saturates the
inequality. This can happen if the probability density and the estimator
satisfy $\frac{\partial}{\partial\theta}\ln f(x;\theta)=k(\theta)\left(\hat{\theta}(x)-\theta\right).$

\subsection{Generalized Cramér-Rao inequalities for higher moments and escort
distributions }

The Fisher information is usually defined as the second order moment
of the score function, the derivative of the log-likelihood, but this
definition can be extended to other moments, leading to a generalized
version of the Cramér-Rao inequality. This extension, which seems
not well known, can be traced back to Barakin \cite[Corollary 5.1]{barankin_locally_1949}.
This generalized Fisher information, together with the extension of
the Cramér-Rao inequality, has also been exhibited by Vajda \cite{vajda_-divergence_1973}
as a limit of a $\chi^{\alpha}$-divergence.  We will recover this
general Cramér-Rao inequality, as well as the standard one, as a particular
case of our new $q$-Cramér-Rao inequalities. The main idea here is
to compute the bias, or a moment of the error, with respect to an
escort distribution $g(x;\theta)$ of $f(x;\theta)$ instead of the
initial distribution. If we first consider the $q$-bias defined by
$B_{q}(\theta):=E_{q}\left[\hat{\theta}(x)-\theta\right]$, we have
the following general statement
\begin{thm}
\label{prop:[Generalized-Cram=0000E9r-Rao-inequali}{[}Generalized
$q$-Cramér-Rao inequality{]} - Let $f(x;\theta)$ be a univariate
probability density function defined over a subset $X$ of $\mathbb{R}$,
and $\theta\in\Theta$ a parameter of the density. Assume that $f(x;\theta)$
is a jointly measurable function of $x$ and $\theta,$ is integrable
with respect to $x$, is absolutely continuous with respect to $\theta,$
and that the derivative with respect to $\theta$ is locally integrable.
Assume also that $q>0$ and that $M_{q}[f;\theta]$ is finite. For
any estimator $\hat{\theta}(x)$ of $\theta$, we have 
\begin{gather}
E\left[\left|\hat{\theta}(x)-\theta\right|^{\alpha}\right]^{\frac{1}{\alpha}}I_{\beta,q}[f;\theta]^{\frac{1}{\beta}}\geq\left|1+\frac{\partial}{\partial\theta}E_{q}\left[\hat{\theta}(x)-\theta\right]\right|\label{eq:GeneralizedCramerRao}
\end{gather}
with $\alpha$ and $\beta$ Hölder conjugates of each other, i.e.
$\alpha^{-1}+\beta^{-1}=1,$ $\alpha\geq1$, and where the quantity
\begin{alignat}{1}
I_{\beta,q}[f;\theta] & =E\left[\left|\frac{f(x;\theta)^{q-1}}{M_{q}[f;\theta]}\,\frac{\partial}{\partial\theta}\ln\left(\frac{f(x;\theta)^{q}}{M_{q}[f;\theta]}\right)\right|^{\beta}\right]=\left(\frac{q}{M_{q}[f;\theta]^{\frac{1}{q}}}\right)^{\beta}E\left[\left|\frac{\partial}{\partial\theta}\ln_{q*}\left(\frac{f(x;\theta)}{M_{q}[f;\theta]^{\frac{1}{q}}}\right)\right|^{\beta}\right]\label{eq:GeneralizedFisher}
\end{alignat}
is the generalized Fisher information of order $(\beta,q)$ on the
parameter $\theta.$ The equality case is obtained if
\begin{equation}
\frac{q}{M_{q}[f;\theta]^{\frac{1}{q}}}\,\frac{\partial}{\partial\theta}\ln_{q*}\left(\frac{f(x;\theta)}{M_{q}[f;\theta]}\right)=c(\theta)\,\mathrm{sign}\left(\hat{\theta}(x)-\theta\right)\left|\hat{\theta}(x)-\theta\right|^{\alpha-1},\label{eq:CaseOfEqualityInCR}
\end{equation}
with $c(\theta)>0.$
\end{thm}
Observe that in the case $q=1,$ $M_{1}[f;\theta]=1$ and the deformed
logarithm reduces to the standard one. Immediately, we obtain the
extended Barakin-Vajda Cramér-Rao inequality in the $q=1$ case, as
well as the standard Cramér-Rao inequality (\ref{eq:StandardCramerRao})
when $q=1$ and $\alpha=\beta=2$. 

\begin{cor}
{[}Barakin-Vajda Cramér-Rao inequality{]} - Under the same hypotheses
as in Theorem\,\ref{prop:[Generalized-Cram=0000E9r-Rao-inequali},
we have 
\begin{equation}
E\left[\left|\hat{\theta}(x)-\theta\right|^{\alpha}\right]^{\frac{1}{\alpha}}I_{\beta,1}[f;\theta]^{\frac{1}{\beta}}\geq\left|1+\frac{\partial}{\partial\theta}E\left[\hat{\theta}(x)-\theta\right]\right|\label{eq:GeneralizedCramerRao_qequal1}
\end{equation}
with 
\begin{equation}
I_{\beta,1}[f;\theta]=E\left[\left|
\frac{\partial}{\partial\theta}\ln\left(f(x;\theta)\right)
\right|^{\beta}
\right]
\end{equation}
and equality if\textup{ $\frac{\partial}{\partial\theta}\ln\left(f(x;\theta)\right)=c(\theta)\mathrm{sign}\left(\hat{\theta}(x)-\theta\right)\left|\hat{\theta}(x)-\theta\right|^{\alpha-1}.$} 
\end{cor}

This inequality generalizes the standard $\alpha=2$ Cramér-Rao inequality
to moments of any order $\alpha>1.$ 
\begin{proof}
{[}of Theorem \ref{prop:[Generalized-Cram=0000E9r-Rao-inequali}{]}
Consider the derivative of the $q$-bias
\begin{equation}
\frac{\partial}{\partial\theta}B_{q}(\theta)=\frac{\partial}{\partial\theta}\int_{X}\left(\hat{\theta}(x)-\theta\right)\frac{f(x;\theta)^{q}}{M_{q}[f;\theta]}\text{d}x.
\end{equation}
The regularity conditions in the statement of the theorem enable to
interchange integration with respect to $x$ and differentiation with
respect to $\theta,$ so that
\begin{align*}
\frac{\partial}{\partial\theta}\int_{X}\left(\hat{\theta}(x)-\theta\right)\frac{f(x;\theta)^{q}}{M_{q}[f;\theta]}\text{d}x & =-\int_{X}\frac{f(x;\theta)^{q}}{M_{q}[f;\theta]}\text{d}x\\
 & +\int_{X}\left(\hat{\theta}(x)-\theta\right)\left[q\frac{\frac{\partial}{\partial\theta}f(x;\theta)}{f(x;\theta)}-\frac{\frac{\partial}{\partial\theta}M_{q}[f;\theta]}{M_{q}[f;\theta]}\right]\frac{f(x;\theta)^{q-1}}{M_{q}[f;\theta)]}f(x;\theta)\,\text{d}x,
\end{align*}
or, since the first term on the right is equal to -1 and since the
term in bracket can be written as the derivative of the logarithm
of the escort distribution of $f(x;\theta)$,
\begin{equation}
1+\frac{\partial}{\partial\theta}B_{q}(\theta)=\int_{X}\left(\hat{\theta}(x)-\theta\right)\,\frac{\partial}{\partial\theta}\ln\left(\frac{f(x;\theta)^{q}}{M_{q}[f;\theta]}\right)\,\frac{f(x;\theta)^{q-1}}{M_{q}[f;\theta]}f(x;\theta)\,\text{d}x.\label{eq:UnPlusdB}
\end{equation}
Consider the absolute value of the integral above, which is less than
the integral of the absolute value of the integrand. By the Hölder
inequality, with $\alpha>1$ and $\beta$ its Hölder conjugate, we
then have
\begin{equation}
\left|1+\frac{\partial}{\partial\theta}B_{q}(\theta)\right|\leq\left(\int_{X}\left|\hat{\theta}(x)-\theta\right|^{\alpha}\, f(x;\theta)\,\text{d}x\right)^{\frac{1}{\alpha}}\,\left(\int_{X}\left|\frac{\partial}{\partial\theta}\ln\left(\frac{f(x;\theta)^{q}}{M_{q}[f;\theta]}\right)\,\frac{f(x;\theta)^{q-1}}{M_{q}[f;\theta]}\right|^{\beta}f(x;\theta)\,\text{d}x\right)^{\frac{1}{\beta}}
\end{equation}
which is the generalized Cramér-Rao inequality (\ref{eq:GeneralizedCramerRao}).
By elementary calculations, we can identify that the generalized Fisher
information above can also be expressed as the derivative of the $q_{*}$-logarithm,
as indicated in the right side of (\ref{eq:GeneralizedFisher}). Finally,
the case of equality follows from the condition of equality in the
Hölder inequality, and from the requirement that the integrand in
(\ref{eq:UnPlusdB}) is non negative: this gives
\begin{equation}
\left|\frac{q}{M_{q}[f;\theta]^{\frac{1}{q}}}\,\frac{\partial}{\partial\theta}\ln_{q*}\left(\frac{f(x;\theta)}{M_{q}[f;\theta]^{\frac{1}{q}}}\right)\right|^{\beta}=k(\theta)\left|\hat{\theta}(x)-\theta\right|^{\alpha}\,\,\,\text{ and \,\,\,}\left(\hat{\theta}(x)-\theta\right)\frac{\partial}{\partial\theta}\ln_{q*}\left(\frac{f(x;\theta)}{M_{q}[f;\theta]^{\frac{1}{q}}}\right)\geq0,
\end{equation}
which can be combined into the single condition (\ref{eq:CaseOfEqualityInCR}),
with $c(\theta)=k(\theta)^{\frac{1}{\beta}}>0.$
\end{proof}
By the properties of escort distributions, we can also obtain an inequality
that involves the $q$-moment of the error $\left(\hat{\theta}(x)-\theta\right)$
instead of the standard moment. Indeed, if $g(x;\theta)$ denotes
the escort distribution of $f(x,\theta)$ of order $q,$ then, as
already mentioned, $f(x;\theta)$ is the escort of order $\bar{q}$
of $g(x;\theta)$, and 
\begin{equation}
f(x;\theta)=\frac{g(x;\theta)^{\bar{q}}}{N_{q}[g;\theta]},
\end{equation}
with $N_{q}[g;\theta]=\int_{X}g(x;\theta)^{\bar{q}}\text{d}x=M_{q}[f;\theta]^{-\bar{q}}.$ 

With these notations, we see that the expectation with respect to
$f$ is the $\bar{q}$-expectation with respect to $g$, and that
the $q$-expectation with respect to $f$ is simply the standard expectation
with respect to $g$. On the other hand, we also have a simple property
that links the deformed logarithms of orders $q_{*}$ and $\bar{q}$:
\begin{prop}
Let $b>0,$ $a=b^{q}.$ With $q_{*}=2-q$ and $\bar{q}=1/q$, the
following equality holds:
\begin{equation}
\ln_{\bar{q}}(a)=q\ln_{q_{*}}(b).\label{eq:relationLnq}
\end{equation}
\end{prop}
\begin{proof}
By direct verification.
\end{proof}
In particular, we note that with $a=g(x;\theta)=b^{q}=\frac{f(x;\theta)^{q}}{M_{q}[f;\theta]},$
we have 
\begin{equation}
\ln_{\bar{q}}(g)=q\ln_{q_{*}}(f(x;\theta)/M_{q}[f;\theta]^{\frac{1}{q}}).
\end{equation}

With these elements, the simple expression of the extended Cramér-Rao
inequality (\ref{eq:GeneralizedCramerRao}) in terms of the escort
$g(x;\theta)$ of $f(x,\theta)$ yields the following corollary. 
\begin{cor}
\label{cor2}{[}Generalized escort-q-Cramér-Rao inequality{]} - Under
the same hypotheses as in Theorem\,\ref{prop:[Generalized-Cram=0000E9r-Rao-inequali},
we have
\begin{gather}
E_{\bar{q}}\left[\left|\hat{\theta}(x)-\theta\right|^{\alpha}\right]^{\frac{1}{\alpha}}\bar{I}_{\beta,\bar{q}}[g;\theta]^{\frac{1}{\beta}}\geq\left|1+\frac{\partial}{\partial\theta}E\left[\hat{\theta}(x)-\theta\right]\right|\label{eq:GeneralizedCramerRao-1}
\end{gather}
with $\alpha$ and $\beta$ Hölder conjugates of each other, i.e.
$\alpha^{-1}+\beta^{-1}=1,$ $\alpha\geq1$, and where the quantity
\begin{equation}
\bar{I}_{\beta,\bar{q}}[g;\theta]=\left(N_{q}[g;\theta]\right)^{\beta}\, E_{\bar{q}}\left[\left|g(x;\theta)^{1-\bar{q}}\,\frac{\partial}{\partial\theta}\ln\left(g\right)\right|^{\beta}\right]=\left(N_{q}[g;\theta]\right)^{\beta}E_{\bar{q}}\left[\left|\frac{\partial}{\partial\theta}\ln_{\bar{q}}\left(g(x;\theta\right)\right|^{\beta}\right]\label{eq:GeneralizedFisher-1}
\end{equation}
is the generalized Fisher information of order $(\beta,q)$ on the
parameter $\theta.$ The equality case is obtained if 
\begin{equation}
\frac{\partial}{\partial\theta}\ln_{\bar{q}}\left(g(x;\theta\right)=c(\theta)\mathrm{sign}\left(\hat{\theta}(x)-\theta\right)\left|\hat{\theta}(x)-\theta\right|^{\alpha-1}.\label{eq:CaseofEqualityInCR2}
\end{equation}

\end{cor}
Note that this is simply a rewriting of the initial extended expression
of extended Cramér-Rao inequality (\ref{eq:GeneralizedCramerRao})
in terms of the escort $g(x;\theta)$ of $f(x,\theta)$. The generalized
Fisher information $\bar{I}_{\beta,\bar{q}}[g]$ is the same as $I_{\beta,q}[f;\theta],$
up to the rewriting in terms of $g$. The second Cramér-Rao inequality
in (\ref{eq:GeneralizedCramerRao-1}) is nice because it exhibits
a fundamental estimation bound for a $q$-moment on the estimation
error, thus making a bridge between concepts in estimation theory
and the tools of nonextensive thermostatistics. What we learn from
this result is the fact that for all estimators with a given bias,
the best estimator that minimizes the $q$-moment of the error is
lower bounded by the inverse of the (generalized) Fisher information.

We shall also discuss in some more details the case of equality in
the two Cramér-Rao inequality. It appears that the general solution
that saturates the bounds is in the form of a deformed $q$-exponential. 

Consider the conditions of equality (\ref{eq:CaseOfEqualityInCR})
and (\ref{eq:CaseofEquality}) in the two Cramér-Rao inequalities.
In the first case, we have that the distribution which attains the
bound shall satisfy
\begin{equation}
\frac{\partial}{\partial\theta}\ln_{q*}\left(\frac{f(x;\theta)}{M_{q}[g;\theta]^{\frac{1}{q}}}\right)=c(\theta)\,\text{sign}\left(\hat{\theta}(x)-\theta\right)\left|\hat{\theta}(x)-\theta\right|^{\alpha-1},\label{eq:CaseofEquality}
\end{equation}
where $c(\theta)$ is a positive function. The general solution of
this differential equation has the form
\begin{equation}
f(x;\theta)\propto\exp_{q_{*}}\left(\int_{\Theta}c(\theta)\,\text{sign}\left(\hat{\theta}(x)-\theta\right)\left|\hat{\theta}(x)-\theta\right|^{\alpha-1}\text{d}\theta\right).\label{eq:SolGeneraleExpq}
\end{equation}
Similarly, in the second Cramér-Rao inequality, we get that
\begin{equation}
g(x;\theta)\propto\exp_{\bar{q}}\left(\int_{\Theta}c(\theta)\,\text{sign}\left(\hat{\theta}(x)-\theta\right)\left|\hat{\theta}(x)-\theta\right|^{\alpha-1}\text{d}\theta\right),\label{eq:SolGeneraleExpq2}
\end{equation}
which is the escort of $f(x;\theta).$

\subsection{Yet another pair of Cramér-Rao inequalities}

It is quite immediate to extend the Cramér-Rao inequalities above
to an even broader context: let us consider a general pair of escort
distributions linked by say $g=\phi(f)$, with $\phi:\mathbb{R}^+\rightarrow\mathbb{R}^+$ monotone increasing, and $f=\phi^{-1}(g)=\psi(g)$.
Denote $E_{\phi}$ and $E_{\psi}$ the corresponding expectations, e.g. $E_{\phi}\left[|x|^\alpha\right]=\int_X |x|^\alpha \phi\left(f(x)\right)\text{d}x$. 
Following the very same steps as in Theorem \ref{prop:[Generalized-Cram=0000E9r-Rao-inequali},
we readily arrive at
\begin{gather}
E\left[\left|\hat{\theta}(x)-\theta\right|^{\alpha}\right]^{\frac{1}{\alpha}}I_{\beta,\phi}[f;\theta]^{\frac{1}{\beta}}\geq\left|1+\frac{\partial}{\partial\theta}E_{\phi}\left[\hat{\theta}(x)-\theta\right]\right|\label{eq:GeneralizedCramerRao-3}
\end{gather}
where 
\begin{equation}
I_{\beta,\phi}[f;\theta]=\int_{X}f(x;\theta) \left|\frac{{\partial\phi(f)}/{\partial\theta}}{f(x;\theta)}\right|^{\beta}\text{d}x=E\left[\left|\frac{{\partial\phi(f)}/{\partial\theta}}{f(x;\theta)}\right|^{\beta}\right].\label{eq:GeneralizedFisher-3}
\end{equation}
Then, the analog of corollary\,\ref{cor2} takes the form 
\begin{equation}
E_{\psi}\left[\left|\hat{\theta}(x)-\theta\right|^{\alpha}\right]^{\frac{1}{\alpha}}\bar{I}_{\beta,\psi}[g;\theta]^{\frac{1}{\beta}}\geq\left|1+\frac{\partial}{\partial\theta}E\left[\hat{\theta}(x)-\theta\right]\right|,\label{eq:GeneralizedCramerRaoNaudts}
\end{equation}
with
\begin{equation}
\bar{I}_{\beta,\psi}[g;\theta]=\int_{X}\psi(g)\left|\frac{\partial g/\partial\theta}{\psi(g)}\right|^{\beta}\text{d}x=E_{\psi}\left[\left|\frac{\partial}{\partial\theta}\ln_{\psi}\left(g\right)\right|^{\beta}\right],\label{eq:GeneralizedFisherNaudts}
\end{equation}
where the function $\ln_{\psi}(u)$ is defined by $\ln_{\psi}(u):=\int_{1}^{u}\frac{1}{\psi(x)}\text{d}x$,
and where the equality case in (\ref{eq:GeneralizedCramerRaoNaudts})
occurs if and only if $g(x;\theta)\propto\exp_{\psi}\int_{\Theta}k(\theta)\left|\hat{\theta}(x)-\theta\right|^{\alpha-1}\text{d}\theta,$
with $\exp_{\psi}$ the inverse function of $\ln_{\psi}.$ During
the writing of this paper, we realized  that a result similar to
(\ref{eq:GeneralizedCramerRaoNaudts}), though obtained using a different
approach and with slightly different notations, has been given in
a deep paper by Naudts \cite{naudts_estimators_2004}. In this interesting
work, the author studied general escort distributions and introduced,
in particular, the notion of $\psi$-exponential families.

\section{\label{sec:Inequalities-in-the}Inequalities in the case of a translation
family}

In the particular case of a translation parameter, our $q$-Cramér-Rao
inequalities reduce to two interesting inequalities that characterize
the $q$-Gaussian distributions. 

Let $\theta$ be a location parameter, and define by $f(x;\theta)$
the family of density $f(x;\theta)=f(x-\theta)$. In this case, we
have that $\frac{\partial}{\partial\theta}f(x;\theta)=-\frac{\mathrm{d} ~}{\mathrm{d} x}f(x-\theta),$
and the Fisher information becomes a characteristic of the information
in the distribution. 

Let us denote by $\mu_{q}$ the $q$-mean of $f(x),$ that is of the
escort distribution $g(x)$ associated with $f(x).$ We immediately
have that the $q$-mean of $f(x;\theta)$ is $(\mu_{q}+\theta)$.
Thus, the estimator $\hat{\theta}(x)=x-\mu_{q}$ is a $q$-unbiased
estimator of $\theta$, since $E_{q}\left[\hat{\theta}(x)-\theta\right]=0$.
Similarly, if we choose $\hat{\theta}(x)=x$, the estimator will be
biased, $E_{q}\left[\hat{\theta}(x)-\theta\right]=\mu_{q}$, but independent
of $\theta$, so that the derivative of the bias with respect to $\theta$
is zero. Finally, let us observe that for a translation family, the
information generating function $M_{q}[f;\theta]=M_{q}[f]$ is independent
of the parameter $\theta$.

\subsection{Cramér-Rao characterizations of generalized $q$-Gaussian distributions}

These simple observations can be applied directly to our two Cramér-Rao
inequalities (\ref{eq:GeneralizedCramerRao}) and (\ref{eq:GeneralizedCramerRao-1}).
This is stated in the two following corollaries. 
\begin{cor}
\label{cor3}{[}Generalized $q$-Cramér-Rao inequality{]} - Let $f(x)$
be a univariate probability density function defined over a subset
$X$ of $\mathbb{R}$. Assume that $f(x)$ is a measurable function
of $x$, is integrable with respect to $x$. Assume also that $q>0$
and that $M_{q}[f]$ is finite. The following generalized Cramér-Rao
inequality then holds: 
\begin{gather}
E\left[\left|x\right|^{\alpha}\right]^{\frac{1}{\alpha}}I_{\beta,q}[f]^{\frac{1}{\beta}}\geq1\label{eq:GeneralizedCramerRao-2}
\end{gather}
with $\alpha$ and $\beta$ Hölder conjugates of each other, i.e.
$\alpha^{-1}+\beta^{-1}=1,$ $\alpha\geq1$, and where the quantity
\begin{equation}
I_{\beta,q}[f]=E\left[\left|\frac{q}{M_{q}[f]}\, f(x)^{q-1}\frac{\mathrm{d} ~}{\mathrm{d} x}\ln\left(f(x)\right)\right|^{\beta}\right]=\left(\frac{q}{M_{q}[f]}\right)^{\beta}E\left[\left|\frac{\mathrm{d} ~}{\mathrm{d} x}\ln_{q*}\left(f(x)\right)\right|^{\beta}\right]\label{eq:GeneralizedFisher-2}
\end{equation}
is the generalized Fisher information of order $(\beta,q)$ of the
distribution. The equality case is obtained if 
\begin{equation}
f(x)\propto\exp_{q_{*}}\left(-\gamma\left|x\right|^{\alpha}\right),\mathrm{\,\, with}\,\,\gamma>0.\label{eq:CaseOfEqualityInCR-1}
\end{equation}
\end{cor}
\begin{proof}
This is a direct consequence of (\ref{eq:GeneralizedCramerRao}),
with $\hat{\theta}(x)=x$ and $\theta=0.$ The case of equality is
obtained by integration and simplifications of (\ref{eq:CaseofEquality}), where the derivative
with respect to $\theta$ is replaced by the derivative with respect
to $x;$ with $\frac{\partial}{\partial\theta}f(x;\theta)=-\frac{\mathrm{d} ~}{\mathrm{d} x}f(x-\theta).$\end{proof}
\begin{cor}
\label{cor4}{[}generalized escort-q-Cramér-Rao inequality{]} - Under
the same hypotheses as in Theorem\,\ref{prop:[Generalized-Cram=0000E9r-Rao-inequali},
we have
\begin{gather}
E_{\bar{q}}\left[\left|x\right|^{\alpha}\right]^{\frac{1}{\alpha}}\bar{I}_{\beta,\bar{q}}[g]^{\frac{1}{\beta}}\geq1\label{eq:GeneralizedCramerRao-1-1}
\end{gather}
with $\alpha$ and $\beta$ Hölder conjugates of each other, i.e.
$\alpha^{-1}+\beta^{-1}=1,$ $\alpha\geq1$, and where the quantity
\begin{equation}
\bar{I}_{\beta,\bar{q}}[g]=\left(N_{q}[g]\right)^{\beta}\, E_{\bar{q}}\left[\left|g(x)^{1-\bar{q}}\,\frac{\mathrm{d} ~}{\mathrm{d} x}\ln\left(g\right)\right|^{\beta}\right]=\left(N_{q}[g]\right)^{\beta}E_{\bar{q}}\left[\left|\frac{\mathrm{d} ~}{\mathrm{d} x}\ln_{\bar{q}}\left(g(x\right)\right|^{\beta}\right]\label{eq:GeneralizedFisher-1-1}
\end{equation}
is the generalized Fisher information of order $(\beta,q)$ of the
distribution $g.$ The equality case is obtained if and only if 
\begin{equation}
g(x)\propto\exp_{\bar{q}}\left(-\gamma\left|x\right|^{\alpha}\right),\text{ with }\gamma>0.
\end{equation}

\end{cor}
In these two cases, the general extended Cramér-Rao inequalities lead
to inequalities for the moments of the distribution, where the equality
is achieved for a generalization of the Gaussian distribution. 

Let us finally note that by the same reasoning as above, the general inequalities (\ref{eq:GeneralizedCramerRao-3}) and (\ref{eq:GeneralizedCramerRaoNaudts}) yield
\begin{equation}
E\left[|x|^\alpha\right]^{\frac{1}{\alpha}} E\left[  \left|\frac{\mathrm{d} ~}{\mathrm{d} x}\ln_{\psi}\left(\phi(f(x))\right)\right|^{\beta} \right]^{\frac{1}{\beta}} \geq 1
\end{equation}
with equality if and only if
\begin{equation}
g(x)=\phi(f(x))\propto\exp_{\psi}\left(-\gamma\left|x\right|^{\alpha}+k\right),\text{ with }\gamma>0 \text{  and } k \text{ a constant}.
\end{equation}

\subsection{Connections with earlier results}

In the case $q=1,$ the characterization result in Corollary\,\ref{cor3}
has first been given by Boekee \cite{boekee_extension_1977}, who
studied the generalized Fisher information $I_{\beta,1}[f]$ and gave
a Cramér-Rao inequality saturated by the generalized Gaussian $g(x)\propto\exp\left(-\gamma\left|x\right|^{\alpha}\right).$ 

It is also important to link our findings to a result by Lutwak \textit{et
al.} \cite{lutwak_Cramer_2005}. In that remarkable paper, the authors
defined a generalized Fisher information, which can be written as
\begin{equation}
\phi_{\beta,q}[f]=E\left[\left|\, f(x)^{q-1}\frac{\mathrm{d} ~}{\mathrm{d} x}\ln\left(f(x)\right)\right|^{\beta}\right]=E\left[\left|\frac{\mathrm{d} ~}{\mathrm{d} x}\ln_{q*}\left(f(x)\right)\right|^{\beta}\right]
\end{equation}
and is similar to our $I_{\beta,q}[f]$ in (\ref{eq:GeneralizedFisher-2}),
up to a factor$\left(\frac{q}{M_{q}[f]}\right)^{\beta}$. Then, they
established a general Cramér-Rao inequality in the form 
\begin{equation}
E\left[\left|x\right|^{\alpha}\right]^{\frac{1}{\alpha}}\phi_{\beta,q}[f]^{\frac{1}{\beta q}}\geq E_{G}\left[\left|x\right|^{\alpha}\right]^{\frac{1}{\alpha}}\phi_{\beta,q}[G]^{\frac{1}{\beta q}},\label{eq:LutwakCRInequality}
\end{equation}
 where $G$ is any generalized Gaussian as in (\ref{eq:qgauss_general}).
Actually, their result (obtained in a very different way), can be
seen as an improved version of (\ref{eq:GeneralizedCramerRao-2}).
Indeed, rewriting the inequality (\ref{eq:GeneralizedCramerRao-2})
in terms of $\phi_{\beta,q}[f],$ we have
\begin{equation}
E\left[\left|x\right|^{\alpha}\right]^{\frac{1}{\alpha}}\phi_{\beta,q}[f]^{\frac{1}{\beta}}\geq q^{-1}M_{q}[f].\label{eq:eqCRavecMq}
\end{equation}
Then, the inequality (\ref{eq:LutwakCRInequality}) can be obtained
by minimizing the lower bound in the right of (\ref{eq:eqCRavecMq}),
as is described in \cite{bercher_betaq-generalized_2012}. 

Similarly, the characterization result in Corollary\,\ref{cor4}
can be connected to a recent result by Furuichi \cite{furuichi_maximum_2009,furuichi_generalized_2010}
in the case $\alpha=\beta=2.$ In these very interesting works, the
author investigated Cramér-Rao inequalities involving $q$-expectations.
More precisely, he considered unnormalized escort distributions, that
is distributions $g(x)=f(x)^{q},$ and defined expectations $\tilde{E}_{q}[.]$
as the expectations computed with respect to these unnormalized escort.
He defined a generalized Fisher information which is essentially the
same as our Fisher information (\ref{eq:GeneralizedFisher-1-1}) --
although it is written in terms of unnormalized $q$-expectation.
Then, he derived a Cramér-Rao inequality \cite[Theorem 1]{furuichi_generalized_2010},
\cite[Theorem 4.1]{furuichi_maximum_2009}, with its case of equality.
This inequality can be recovered at once from (\ref{eq:GeneralizedCramerRao-1-1}),
which is rewritten below in a developed form

\begin{equation}
\left(\int_{X}\frac{g(x)^{\bar{q}}}{N_{q}[g]}\,\left|x\right|^{\alpha}\text{d}x\right)^{\frac{1}{\alpha}}\times N_{q}[g]\left(\int_{X}\frac{g(x)^{\bar{q}}}{N_{q}[g]}\,\left|\frac{\mathrm{d} ~}{\mathrm{d} x}\ln_{\bar{q}}\left(g(x\right)\right|^{\beta}\text{d}x\right)^{\frac{1}{\beta}}\,\geq1.
\end{equation}
It suffices to simplify the normalizations $N_{q}[g],$ using the
fact that $\alpha^{-1}+\beta^{-1}=1$ to get

\begin{gather}
\tilde{E}_{\bar{q}}\left[\left|x\right|^{\alpha}\right]^{\frac{1}{\alpha}}\tilde{E}_{\bar{q}}\left[\left|\frac{\mathrm{d} ~}{\mathrm{d} x}\ln_{\bar{q}}\left(g(x\right)\right|^{\beta}\right]\geq1,\label{eq:GeneralizedCramerRaoFuruichi}
\end{gather}
recovering Furuichi's definition of generalized Fisher information
and the associated Cramér-Rao inequality, with equality if and only
if $g(x)\propto\exp_{\bar{q}}(-\gamma |x|^{\alpha}).$

\section{\label{sec:Further-remarks}Further remarks}

In this section, we add some further comments on two possible estimation
procedures that can be derived by examination of the condition of
equality in the $q$-Cramér-Rao inequalities.

\subsection{Maximum escort likelihood}

Let us first return to the case of equality in the generalized $q$-Cramér-Rao
inequalities. For the second Cramér-Rao inequality, the condition
(\ref{eq:CaseofEqualityInCR2}) is 
\begin{equation}
\frac{\partial}{\partial\theta}\ln_{\bar{q}}\left(g(x;\theta)\right)=c(\theta)\,\text{sign}\left(\hat{\theta}(x)-\theta\right)\left|\hat{\theta}(x)-\theta\right|^{\alpha-1}.\label{eq:ConditionEquality}
\end{equation}
Thus, we see that if the bound is attained (the estimator could then
be termed ``$q$-efficient''), then this suggests to look for the
parameter that maximizes the escort distribution of the likelihood:
\begin{equation}
\hat{\theta}_{MEL}=\arg\max_{\theta}\, g(x;\theta)=\arg\max_{\theta}\frac{f(x;\theta)^{q}}{M_{q}[f(x;\theta)]},\label{eq:MaximumEscort}
\end{equation}
where MEL stands for ``maximum escort likelihood''. Indeed, in these
conditions, we have that the derivative in the left of (\ref{eq:ConditionEquality})
is zero, and thus that
\begin{equation}
\left.\frac{\partial}{\partial\theta}\ln_{\bar{q}}\left(g(x;\theta)\right)\right|_{\theta=\hat{\theta}_{MEL}}=\left.c(\theta)\,\text{sign}\left(\hat{\theta}(x)-\theta\right)\left|\hat{\theta}(x)-\theta\right|^{\alpha-1}\right|_{\theta=\hat{\theta}_{MEL}}=0.\label{eq:DerivateConditionEquality}
\end{equation}
Therefore, we get from the equality in the right side that $\hat{\theta}(x)=\hat{\theta}_{MEL}.$
Hence, we see that if it exists a $q$-efficient estimator, it is
the estimator defined by the maximum of the escort of the likelihood.
Of course, we recover the standard maximum likelihood estimator in
the $q=1$ case. The analysis of the properties of this estimator
will be the subject of future efforts.

\subsection{Maximum Lq-likelihood estimation}

We also saw that in case of equality, then the likelihood, or equivalently
its escort, must be under the form of a $q$-exponential:
\begin{equation}
f(x;\theta)\propto\exp_{q_{*}}\left(\int_{\Theta}c(\theta)\,\text{sign}\left(\hat{\theta}(x)-\theta\right)\left|\hat{\theta}(x)-\theta\right|^{\alpha-1}\text{d}\theta\right).
\end{equation}
 Actually, it seems that there is only some very particular cases
where this could occur. For instance, if the measurements consists
in a series of independent and identically distributed observations
$x_{i},$ then $f(x;\theta)=\Pi_{i}$$f(x_{i};\theta)$ and one would
have to find a distribution $f(x_{i};\theta)$ such that the product
$f(x;\theta)$ writes as a $q_{*}$-exponential. A possible amendment
to the formulation can be to consider a $q_{*}$-product of the densities
$f(x_{i};\theta)$ instead of the standard product, and define $f^{(q_{*})}(x;\theta)=\bigotimes_{q_{*},i}f(x_{i};\theta).$
Such $q$-likelihood has already been considered by \cite{suyari_law_2005}.
Here, the Cramér-Rao inequality still applies for the $q_{*}$-likelihood
$f^{(q_{*})}(x;\theta)$, and the equality is obtained if $f^{(q_{*})}(x;\theta)$
is a $q_{*}$-exponential. By the properties (\ref{eq:qproductidentities})
of the $q$-product, we see that the individual densities $f(x_{i};\theta)$
must be $q_{*}$-exponentials. Similarly, we see that the escort-likelihood
will have the form of $q$-exponential if we use the $q$-product
of the escort densities: $g^{(\bar{q})}(x;\theta)=\bigotimes_{\bar{q},i}\, g(x_{i};\theta).$ 

The equality condition (\ref{eq:CaseofEqualityInCR2}), applied to
the $\bar{q}$ escort-likelihood $g^{(\bar{q})}(x;\theta)$ then suggests
to define the estimator as the maximizer of the $\bar{q}$ escort-likelihood,
or, equivalently, as the maximizer of the $\ln_{\bar{q}}$ escort-likelihood:
\begin{equation}
\hat{\theta}_{MLq}=\arg\max_{\theta}\ln_{\bar{q}}\left(g^{(\bar{q})}(x;\theta)\right)=\arg\max_{\theta}\,\,\left(\sum_{i}\ln_{\bar{q}}g(x_{i};\theta)\right).\label{eq:38}
\end{equation}
Actually, the rule defined by (\ref{eq:38}) has been proposed and
studied in the literature. It has been introduced by Ferrari \cite{ferrari_maximum_2010,ferrari_robust_2011},
and independently by Hasegawa \cite{hasegawa_properties_2009}. As a matter
of fact, the first authors, defining a problem as in (\ref{eq:38})
with data distributed according to $f(x;\theta)$, showed that the
distribution $g(x;\theta)$ must be the escort of $f(x;\theta).$
These authors have shown that (\ref{eq:38}) yields a robust estimator
with a tuning parameter, $q,$ which balances efficiency and robustness.
When the number of data increases, then the estimator appears to be
the empirical version of
\begin{equation}
\hat{\theta}_{MLq}=\arg\min_{\theta}\,\,-E\left[\ln_{\bar{q}}g(x_{i};\theta)\right]=\arg\min_{\theta}\frac{1}{N_{q}[g;\theta]}\,\frac{1}{(1-\bar{q})}\left(\int g(x;\theta)^{\bar{q}}\text{d}x-1\right),
\end{equation}
which is nothing but the normalized Tsallis entropy attached to the
escort distribution $g$. Such links between maximum likelihood and
the minimization of the entropy with respect to the parameter $\theta$
can be traced back to Akaike in \cite{akaike_information_1973}. Here, this
gives a direct interpretation of the MLq method as an approximate
minimum entropy procedure, and highlights the particular role of escort
distributions in this context. Interestingly, it is shown in \cite{ferrari_maximum_2010}
and \cite{hasegawa_properties_2009} that the asymptotical behaviour
of the estimator is governed by a generalized Fisher information similar
to (\ref{eq:GeneralizedFisher-1}). Our findings add the fact that
the MLq estimator satisfies the Cramér-Rao inequality (\ref{eq:GeneralizedCramerRao-1}),
for the product distribution $g^{(\bar{q})}(x;\theta).$

\section{Conclusions}

The generalized $q$-Gaussians form an important and versatile family
of probability distributions. These generalized $q$-Gaussians, which
appear in physical problems as well as in functional inequalities,
are the maximum entropy distributions associated with Tsallis or Rényi
entropy. In this paper, we have shown that the generalized $q$-Gaussians
are also the minimizers of extended versions of the Fisher information,
over all distributions with a given moment, just as the standard
Gaussian minimizes Fisher information over all distributions with
a given variance. Actually, we obtain more precise results in the
form of extended versions of the standard Cramér-Rao inequality, which
are saturated by the generalized $q$-Gaussians. These Cramér-Rao
inequalities, and the associated generalized Fisher informations,
recover, put in perspective and connect earlier results by Lutwak
et al. \cite{lutwak_Cramer_2005} and Furuichi \cite{furuichi_maximum_2009,furuichi_generalized_2010}. 

As a matter of fact, these characterizations of the generalized $q$-Gaussians
appear as simple consequences of more general extended Cramér-Rao
inequalities obtained in the context of estimation theory. Indeed,
considering moments of any order of the estimation error, and using
statistical expectations with respect to an escort distribution, we
have derived two general Cramér-Rao inequalities that still include
the Barakin-Vajda as well as the standard Cramér-Rao inequality as
particular cases. This gives rise to general definitions of generalized
Fisher information, which reduce to the standard one as a particular
case, and make sense in this context. We have also characterized the
case of equality and shown that the lower bounds of the inequalities
can be attained if the parametric density belongs to a $q$-exponential
family. Finally, we have indicated that these findings suggest some
new estimation procedures, recovering in particular a recent Maximum
L$_{q}$-likelihood procedure. 

These results have been derived and presented in the monodimensional
case. An important point will be to extend these results to the multidimensional
case. This would be important both for the estimation inequalities
as well as for the Cramér-Rao inequalities characterizing the generalized
$q$-Gaussians. As far as the latter point is concerned, some results
are already available in \cite{bercher_betaq-generalized_2012},
and should be connected to estimation problems. As is well-known,
the Weyl-Heisenberg uncertainty principle in statistical physics is
nothing but the standard Cramér-Rao inequality for the location parameter.
Thus it would be of particular interest to investigate on the possible
meanings of the uncertainty relationships that could be associated
to the extended Cramér-Rao inequalities. Fisher information,
Cramér-Rao planes have been identified as useful and versatile tools
for characterizing complex systems, see e.g. \cite{dehesa_Cramer-Rao_2006,dehesa_fisher_2006,dehesa_generalized_????},
and it would be therefore interesting to look at the potential benefits
of using the extended versions in such problems. An open issue is
the possible convexity property of the generalized Fisher information.
Indeed, it is known that the standard Fisher information, as well
as the generalized versions with $q=1$, are convex functions of the
density. If this were also true for any value of $q$, then it would
be possible to associate to the generalized Fisher information a statistical
mechanics with the standard Legendre structure and with the $q$-Gaussian
as canonical distribution. Finally, future work shall also examine
the estimation rules suggested by our setting and study their statistical
properties.

\section*{References}


\end{document}